\documentclass[journal,onecolumn]{IEEEtran}
\usepackage{cite}
\usepackage{amsmath,amssymb,amsfonts}
\usepackage{graphicx}
\usepackage{textcomp}
\usepackage{color}
\usepackage{xcolor}
\usepackage{amssymb}
\usepackage{amsmath}
\usepackage{amsthm}
\usepackage{multirow}
\usepackage{multicol,lipsum}
\usepackage[utf8x]{inputenc}
\usepackage{graphicx}
\usepackage{epsfig}
\usepackage{algorithm}
\usepackage{algorithmic}
\usepackage{multirow}
\usepackage{amsmath}
\usepackage{hyperref}
\usepackage{mathtools}

\usepackage{amsfonts}
\usepackage{amssymb}
\usepackage[justification=centering]{caption}
\usepackage{bm}
\usepackage{listings}
\newtheorem{remark}{Remark}
\newtheorem{definition}{Definition}
\newtheorem{theorem}{Theorem}
\newtheorem{proposition}{Proposition}
\newtheorem{lemma}{Lemma}
\newtheorem{example}{Example}
\newtheorem{corollary}{Corollary}
\newtheorem{assumption}{Assumption}
\usepackage{balance}

\newcommand{\oomit}[1]{}

\begin{document}

\title{Robust Invariant Sets Computation for Switched Discrete-Time Polynomial Systems}

\author{Bai Xue and Naijun Zhan}

\maketitle

\begin{abstract}
In this paper we study the robust invariant sets generation problem for discrete-time switched polynomial systems subject to disturbance inputs within the optimal control framework. A robust invariant set of interest  is a set of states such that every possible trajectory starting from it never leaves a specified safe set, regardless of actual disturbances. The maximal robust invariant set is shown to be the zero level set of the unique bounded solution to a Bellman type equation, which is a functional equation being widely used in discrete-time optimal control. This is the main contribution of this work. The uniqueness of bounded solutions enables us to solve the derived Bellman type equation using numerical methods such as the value iteration, which provides an approximation of the maximal robust invariant set. In order to increase the scalability of the Bellman equation based method, a semi-definite program, which is constructed based on the derived Bellman type equation, is also implemented to synthesize robust invariant sets. Finally, three examples demonstrate the performance of our methods.
\end{abstract}

\section{Introduction}
\label{Int}
The computation of robust invariant sets is central to the validation of systems such as nonlinear dynamical systems and hybrid-state extensions thereof \cite{khalil2002nonlinear,lunze2009handbook}. A robust invariant set of interest in this paper refers to a set of states such that every possible trajectory initialized in it never violates a specified safe state constraint irrespective of the actual disturbance. It has many other names in the literature, e.g., reachability tubes \cite{rakovic2008} and invariance kernels \cite{aubin2011}. Due to its widespread applications in systems and control for stability analysis and control design \cite{blanchini2008set}, robust invariant sets generation has been the subject of extensive research over past several decades, e.g.,
\cite{bertsekas1972,rakovic2004,grieder2005,maidens2013,trodden2016,li2018invariance,aubin2011}.

One popular line for studying robust invariant sets generation problem is by exploiting the link to optimal control problems. When the system is continuous-time, Hamilton-Jacobi equations, which are widely used in optimal control theory \cite{Bardi1997}, are explored for performing reachability analysis, e.g., \cite{aubin2011,margellos2011,mitchell2005}. 
From a theoretical point of view, one advantage of this method is that the exact reach sets of interest could be characterized by level sets of viscosity solutions to Hamilton-Jacobi equations. From a computational point of view, there exists well-developed numerical methods \cite{falcone2016,mitchell2007,rakovic2008,aguilar2014} for solving Hamilton-Jacobi equations with appropriate number of state variables, rendering possible the gain of exact reach sets. Recently, \cite{WZF19} proposed a Hamilton-Jacobi equation and characterized the maximal robust invariant set as the zero level set of the unique Lipschitz continuous viscosity solution to this Hamilton-Jacobi equation for state-constrained perturbed continuous-time systems.

Despite significant progress towards the computation of robust invariant sets for continuous-time systems within the optimal control framework, works on the computation of robust invariant sets for its counterpart, i.e. discrete-time systems, are relatively rare in this regard, especially for discrete-time switched nonlinear systems subject to disturbance inputs.
A discrete-time switched system is defined by a family of discrete-time subsystems and a switching rule orchestrating the switching between subsystems \cite{lunze2009handbook}.  Nonlinear discrete-time systems are widespread in many practical systems such as biological systems and economic systems, where the underlying models are in discrete-time \cite{mahmoud2012discrete}. Moreover, a growing number of digital devices are being used for information processing and control purposes in a variety of complex systems applications, including industrial processes, power networks and communication networks. For these applications, it is reasonable to model the system using a discrete-time nonlinear state space model \cite{rabbath2013discrete}. Thus, the study of robust invariant sets for discrete-time nonlinear systems is conducive to improving these applications, as mentioned before.

In this paper we study the generation of robust invariant sets for  discrete-time switched polynomial systems subject to disturbance inputs in the framework of optimal control. We firstly define a bounded value function with a positive-valued parameter falling between zero and one  such that its zero sub-level set is equal to the maximal robust invariant set. Then the value function is reduced to a bounded solution to a Bellman type equation. The Bellman type equation, named after Richard Bellman, is a functional equation being widely used in discrete-time optimal control \cite{Bardi1997}. When the value of the positive-valued parameter is strictly less than one, the bounded solution to the Bellman type equation is unique. The value iteration is shown to be capable of solving the  equation with an appropriate number of state variables, thereby enabling the gain of an estimation of the maximal robust invariant set. When the parameter is equal to one, the resulting Bellman type equation does not feature this uniqueness property. For this case, the Bellman type equation is relaxed into a system of inequalities such that a semi-definite program is constructed for synthesizing robust invariant sets. Finally, three examples demonstrate the performance of our approaches.  

The main contribution of this paper is summarized as follows: The maximal robust invariant set of  a discrete-time switched polynomial system subject to disturbance inputs is characterized as the zero level set of the unique bounded solution to a Bellman type equation. The well-known value iteration is shown to be capable of resolving this equation for obtaining an estimation of the maximal robust invariant set. To the best of our knowledge this is the first work on linking robust invariant sets with Bellman type equations exhibiting unique bounded solutions. 

\subsection*{Related Work}
Most of existing works on the computation of robust invariant sets for discrete-time systems focus on linear systems, e.g. \cite{kolmanovsky1998,rakovic2004,grieder2005,rakovic2005,blanchini2008set,athanasopoulos2010,tahir2012,trodden2016}. 

In the nonlinear case a common practice is to exploit the connection between invariance and Lyapunov functions, and define invariant sets as sub-level sets of Lyapunov functions, e.g.,  \cite{coutinho2013local,giesl2014computation,luk2015}.  Unfortunately, it is  known that the construction of Lyapunov functions for nonlinear systems is a very difficult problem in general \cite{giesl2015review}, although powerful sum-of-squares techniques are proposed and widely used in the past decade to compute Lyapunov functions. When it returns to the computation of  invariant sets, the sum-of-squares technique often leads to nonconvex optimization problems \cite{she2013computing}. Existing Lyapunov-based methods cannot guarantee the calculation of the maximal robust invariant set generally. In our method, the Bellman type equation from discrete-time optimal control is the main mathematical tool. The maximal robust invariant set is characterized as the zero level set of its unique bounded solution, which can be approximated uniformly by the value iteration. 

Another method for invariant sets synthesis is fixed-point methods, e.g., \cite{bertsekas1972}. Numerical implementation of fixed point methods is nontrivial even for linear systems because invariant sets are not guaranteed to be finitely determined \cite{rungger2017},
i.e., computation may not terminate in a finite number of steps. Moreover, in fixed point methods ellipsoidal and polyhedral sets received special attention in the literature since both ellipsoidal and polyhedral sets are convex and hence technically more appealing. There is, however, one major drawback in computing invariant sets based on ellipsoidal or polyhedral sets. The shape of such sets has to be preserved during the set computation. While this is possible for linear systems, it is general not the case for nonlinear systems. In order to overcome these shortcomings, based on branch-and-bound scheme interval analysis approaches were proposed to synthesize robust invariant sets in \cite{li2018invariance,li2018}. 

The work mostly close to the present one is \cite{xue2018}, in which the robust non-termination set is equivalent to the robust invariant set in this paper when the computer program in \cite{xue2018} is transformed into a discrete-time switched system. The main difference between this work and \cite{xue2018} is that the derived Bellman type equation in this work has a unique bounded solution, thereby facilitating the use of existing numerical methods such as the value iteration to estimate the maximal robust invariant set. However, the Bellman type equation in \cite{xue2018} does not feature this uniqueness property.  A resulting consequence is that  numerical methods for solving  the Bellman type equation in \cite{xue2018} cannot guarantee to compute an estimation of the maximal robust invariant set. The work \cite{xue2020cha} presents a Bellman type equation for computing the maximal robust region of attraction. The maximal robust region of attraction is a set of all states such that every trajectory initialized in it will approach an equilibrium  while never violating a specified state constraint, regardless of the actual disturbance. It differs from the maximal robust invariant set considered in this paper in that the maximal robust invariant set is not required to contain equilibria and trajectories starting from it are not required to approach an equilibrium. Recently, a Bellman equation based method was proposed for computing finite time horizon maximal invariant sets in \cite{jonesgeneralization}. Unlike \cite{jonesgeneralization}, the present work focuses on generating the maximal robust invariant set over the infinite time horizon rather than finite time horizons.

\textit{Organization of the paper.} This paper is structured as follows. In Section \ref{Pre}, basic notions used throughout this paper and the problem of interest are introduced. Then we elucidate our approach for synthesizing robust invariant sets in Section \ref{IG}. After demonstrating our approach on three examples in Section \ref{ex}, we  conclude this paper in Section \ref{con}.

\section{Preliminaries}
\label{Pre}
In this section we describe discrete-time switched polynomial systems subject to disturbance inputs, and robust invariant sets of interest in this paper.

\subsection{Problem Description}
\label{PD}
Before posing the problem studied, let us introduce some basic notions used in the rest of this paper: $\mathbb{N}$ stands for the set of nonnegative integers and $\mathbb{R}$ for
the set of real numbers; $\mathbb{R}[\cdot]$ denotes the ring of polynomials in variables given by the argument; $\mathbb{R}_i[\cdot]$ denotes the set of polynomials of degree $i$ in variables given by the argument, $i\in\mathbb{N}$; Vectors are denoted by boldface letters.

We consider discrete-time switched polynomial systems subject to disturbance inputs presented in Definition \ref{ps}.
\begin{definition}
\label{ps}
A discrete-time switched polynomial system subject to  disturbance inputs  (abbr., \emph{DSPS})  is a quintuple $(\bm{x}_0,L,\mathcal{X},\mathcal{D},\mathcal{F})$ with the following components:
\begin{enumerate}
\item[-] $\bm{x}_0\in \mathbb{R}^n$ is the initial state;
\item[-] $L=\{1,\ldots,k\}$ is a finite set of locations;
\item[-] $\mathcal{X}=\{X_i,i=1,\ldots,k\}$ with $X_i\cap X_j=\emptyset$ if $i\neq j$ and $\cup_{i=1}^k X_i=\mathbb{R}^n$, where $X_i$ contains all possible states while at location $i$;
\item[-] $\mathcal{D}=\{D_i,i=1,\ldots,k\}$, where $D_i\subseteq \mathbb{R}^{m_i}$ is the set of disturbance inputs while at location $i$;
\item[-] $\mathcal{F}=\{\bm{f}_i(\bm{x},\bm{d}): X_i\times D_i \rightarrow \mathbb{R}^n,i=1,\ldots,k\}$ with $\bm{f}_i(\bm{x},\bm{d})$ constraining the evolution of the state by the difference equation $\tilde{\bm{x}}=\bm{f}_i(\bm{x},\bm{d})$ while at location $i$,
\end{enumerate}
where  $X_i=\{\bm{x}\in \mathbb{R}^n\mid \bigwedge_{j=1}^{n_i}h_{i,j}(\bm{x})\rhd 0\}$ with $\rhd\in \{\leq, <\}$,  and $D_i=\{\bm{d}\in \mathbb{R}^{m_i}\mid \bigwedge_{j=1}^{m'_{i}}p_{i,j}(\bm{d})\leq 0\}$, $i=1,\ldots,k$. Also, $h_{i,j}(\bm{x})\in \mathbb{R}[\bm{x}]$, $i=1,\ldots,k$, $j=1,\ldots,n_i$; $p_{i,j}(\bm{d})\in \mathbb{R}[\bm{d}]$, $i=1,\ldots,k$, $j=1,\ldots,m'_{i}$; $\bm{f}_i(\bm{x},\bm{d})\in \mathbb{R}[\bm{x},\bm{d}]: X_i \times D_i \rightarrow \mathbb{R}^n$, $i=1,\ldots,k$.
\end{definition}

Prior to presenting the concept of the maximal robust invariant set for  DSPS, we firstly define a disturbance policy  controlling its trajectory.
\begin{definition}
A disturbance policy for  \emph{DSPS} is an ordered sequence $(\bm{d}(l))_{l \in \mathbb{N}}$, where $\bm{d}(\cdot): \mathbb{N}\rightarrow D_j$ with $j\in \{1,\ldots,k\}$. 
\end{definition}

\begin{definition}
\label{tra}
Given a disturbance policy $\pi=(\bm{d}(l))_{l\in \mathbb{N}}$ and an initial state $\bm{x}_0\in \mathbb{R}^n$, if there exists a sequence $(\bm{x}(l))_{l\in \mathbb{N}}$ satisfying 
\begin{equation}
\label{path}
\bm{x}(l+1)=\bm{f}(\bm{x}(l),\bm{d}(l)),
\end{equation}
where $\bm{x}(0)=\bm{x}_0$ and
\begin{equation}
\label{path_1}  
  \bm{f}(\bm{x},\bm{d})=1_{X_1}(\bm{x})\cdot\bm{f}_1(\bm{x},\bm{d})+\cdots+1_{X_k}(\bm{x})\cdot\bm{f}_k(\bm{x},\bm{d})
  \end{equation}
 with 
\begin{equation} 
\label{D}
 \bm{f}(\cdot,\cdot):S\times D\rightarrow \mathbb{R}^n, S=X_i \text{~and~} D=D_i \text{~if~} \bm{x}\in X_i,
 \end{equation}
  and $1_{X_i}(\cdot): X_i\rightarrow \{0,1\}$, $i=1,\ldots,k$, representing the indicator function of the set $X_i$,  i.e.
\[1_{X_i}(\bm{x}):=\begin{cases}
   1, \quad \text{if }\bm{x}\in X_i,\\
   0, \quad \text{if }\bm{x}\notin X_i
\end{cases},\]
then $\pi$ is an admissible disturbance policy for the initial state $\bm{x}_0$. The trajectory $\bm{\phi}_{\bm{x}_0}^{\pi}(\cdot):\mathbb{N}\rightarrow \mathbb{R}^n$, induced by $\pi$ and $\bm{x}_0$, to \emph{DSPS} is then defined as follows:
\[\bm{\phi}_{\bm{x}_0}^{\pi}(l):=\bm{x}(l),\forall l\in \mathbb{N}.\]
Furthermore, we define $\Pi_{\bm{x}_0}$ as the set of admissible disturbance policies for the initial state $\bm{x}_0$.
\end{definition}

Now, we define the maximal robust invariant set such that DSPS starting from it never leaves the safe set  $X=\{\bm{x}\in \mathbb{R}^n\mid \bigwedge_{i=1}^{n_{0}}h_{i}(\bm{x})\leq 0\}$, where $h_{i}(\bm{x})\in \mathbb{R}[\bm{x}]$, $i=1,\ldots,n_0$.
\begin{definition}[Maximal Robust Invariant Set]\label{RNS}
The maximal robust invariant set $\mathcal{R}_0$ is the set of all states in the safe set $X$ such that every possible trajectory of \emph{DSPS} starting from it never leaves $X$, i.e.
\begin{equation}
\mathcal{R}_0=\{\bm{x}_0 \in \mathbb{R}^n\mid \bm{\phi}_{\bm{x}_0}^{\pi}(l)\in X, \forall l\in \mathbb{N}, \forall \pi \in \Pi_{\bm{x}_0}\}.
\end{equation}
A robust invariant set is a subset of the maximal robust invariant set $\mathcal{R}_0$.
\end{definition}

It is worth remarking here that a robust invariant set in Definition \ref{RNS} is not an invariant set in the classical sense. Unlike invariant sets in the classical sense in which trajectories get trapped, it is a set of states such that every possible trajectory of DSPS starting from it does not leave the safe set $X$. Thus, trajectories starting from a robust invariant set in Definition \ref{RNS} may leave it but cannot leave $X$.

The focus of this paper is on the computation of robust invariant sets. In the rest of this paper, we assume that the interior of the maximal robust invariant set $\mathcal{R}_0$ exists for DSPS.
\begin{assumption}
$\mathcal{R}_0^{\circ}\neq \emptyset$, where $\mathcal{R}_0^{\circ}$ is the interior of the maximal robust invariant set $\mathcal{R}_0$.
\end{assumption}

\section{Computing Robust Invariant Sets}
\label{IG}
In this section we elucidate our approach of synthesizing robust invariant sets for DSPS.  Subsection \ref{charac} characterizes the maximal robust invariant set $\mathcal{R}_0$ as the zero level set of the unique bounded solution to a Bellman type equation, which can be solved by the value iteration.  This is the main contribution of this section. Subsection \ref{AA} presents a computationally tractable semi-definite programming problem, which is derived from the Bellman type equation, for synthesizing robust invariant sets. 

\subsection{Bellman Equations }
\label{charac}
In this subsection we reduce the maximal robust invariant set $\mathcal{R}_0$ as the zero (sub)level set of a bounded solution $v(\bm{x}): \mathbb{R}^n\rightarrow \mathbb{R}$ to a Bellman type equation of the following form:
\begin{equation}
\label{eq}
\begin{split}
\min\big\{\inf_{\bm{d}\in D}(&v(\bm{x})-\alpha v(\bm{f}(\bm{x},\bm{d}))),v(\bm{x})-\max_{j\in \{1,\ldots,n_0\}} h'_{j}(\bm{x})\big\}=0,
\end{split}
\end{equation}
where the set $D$ is defined in \eqref{D}, $h'_{j}(\bm{x})=\frac{h_{j}(\bm{x})}{1+h_{j}^2(\bm{x})}$ and $\alpha \in (0,1]$ is a user-specified constant. When $\alpha\in (0,1)$, the bounded solution to \eqref{eq} is unique. The proofs of the aforementioned claims will be presented afterwards. The uniqueness of bounded solutions facilitates the use of the value iteration  in reinforcement learning to solve \eqref{eq}. A question arises naturally: \textit{how can the solution to the equation \eqref{eq} be obtained? } 

In order to solve the above mentioned question, we first present a value function constructed based on a scalar value $\alpha\in (0,1]$ and the family of functions $(h_j(\bm{x}))_{j=1}^{n_0}$ defining the safe set $X$.  Then, based on its underlying property of satisfying the dynamic programming principle as shown in Lemma \ref{dp}, this value function finally boils down to the unique bounded solution to \eqref{eq} when $\alpha \in (0,1)$, which is formally formulated in Theorem \ref{equations}.

The value function $V:\mathbb{R}^n\rightarrow \mathbb{R}$ is defined by:
\begin{equation}
\label{vf}
V(\bm{x}):=\sup_{ \pi \in \Pi_{\bm{x}}}\sup_{l\in \mathbb{N}}\max_{j\in \{1,\ldots,n_0\}}\big\{\alpha^l h'_{j}(\bm{\phi}_{\bm{x}}^{\pi}(l))\big\},
\end{equation}
where $\alpha\in (0,1]$ is a user-specified constant scalar value. Obviously, $-1<h'_{j}(\bm{x})<1$ over $\bm{x}\in \mathbb{R}^n$ for $j=1,\ldots,n_0$. Consequently, $-1\leq V(\bm{x})\leq 1$ holds for $\bm{x}\in \mathbb{R}^n$. 

Unlike  \cite{xue2018}, we introduce a positive valued parameter $\alpha$ and new bounded functions $(h'_{j}(\bm{x}))_{j=1}^{n_0}$ into the construction of the value function \eqref{vf}.  This  enables us to reduce $V(\bm{x})$ to the unique bounded solution to \eqref{eq} when $\alpha\in (0,1)$.  When $\alpha=1$,  \eqref{vf} is just a bounded solution to  \eqref{eq}.  As discussed in \cite{xue2018}, the uniqueness of bounded solutions to \eqref{eq} with $\alpha=1$ cannot be guaranteed and thus we cannot obtain the maximal robust invariant set via solving \eqref{eq} generally. However, it facilitates the computation of robust invariant sets via solving a semi-definite program, as shown in Subsection \ref{AA}.

The following proposition shows the relationship between the value function $V(\bm{x})$ and the maximal robust invariant set $\mathcal{R}_0$. The zero level set of $V(\bm{x})$ is equal to the maximal robust invariant set $\mathcal{R}_0$ when $\alpha\in (0,1)$, and the zero sublevel set of $V(\bm{x})$ is equal to $\mathcal{R}_0$ when $\alpha=1$.  
\begin{proposition}
\label{sets}
$\mathcal{R}_0=\{\bm{x}\in \mathbb{R}^n\mid V(\bm{x})\leq 0\}$, where $\mathcal{R}_0$ is the maximal robust invariant set in Definition \ref{RNS}. Furthermore, when $\alpha\in (0,1)$, $V(\bm{x})\geq 0$ for $\bm{x}\in \mathbb{R}^n$ and, thus, $\mathcal{R}_0=\{\bm{x}\in \mathbb{R}^n\mid V(\bm{x})=0\}$.
\end{proposition}
\begin{proof}
Let $\bm{y}\in \mathcal{R}_0$.  According to Definition~\ref{RNS}, we have 
\begin{equation}
\label{ineq}
h_{j}(\bm{\phi}_{\bm{y}}^{\pi}(i)) \leq 0, \forall i\in \mathbb{N}, \forall \pi \in \Pi_{\bm{y}}, \forall j\in \{1,\ldots,n_0\}    
\end{equation}
holds, implying that $$h'_{j}(\bm{\phi}_{\bm{y}}^{\pi}(i)) \leq 0, \forall i\in \mathbb{N}, \forall \pi \in \Pi_{\bm{y}}, \forall j\in \{1,\ldots,n_0\}$$ and, thus, $V(\bm{y})\leq 0$.
Therefore, $\bm{y}\in \{\bm{x}\mid V(\bm{x})\leq 0\}$.

On the other hand, if $\bm{y}\in \{\bm{x}\in \mathbb{R}^n\mid V(\bm{x})\leq 0\}$, then $V(\bm{y})\leq 0$, implying that \eqref{ineq} holds and consequently $\bm{y}\in \mathcal{R}_0$. Therefore, $\mathcal{R}_0=\{\bm{x}\in \mathbb{R}^n\mid V(\bm{x})\leq 0\}$. 

As to the case of $\alpha\in (0,1)$, it is obvious that $V(\bm{x})\geq 0$ for $\bm{x}\in \mathbb{R}^n$ since $$\lim_{l\rightarrow\infty}\alpha^l \max_{j\in\{1,\ldots,n_0\}}\{h'_{j}(\bm{\phi}_{\bm{x}}^{\pi}(l))\}=0, \forall \bm{x} \in \mathbb{R}^n, \forall \pi \in \Pi_{\bm{x}}.$$ Therefore, $\mathcal{R}_0=\{\bm{x}\in \mathbb{R}^n\mid V(\bm{x})=0\} \text{~if~} \alpha\in (0,1).$
\end{proof}

Proposition \ref{sets} tells us that the maximal robust invariant set $\mathcal{R}_0$ can be obtained if the value function $V(\bm{x})$ in \eqref{vf} is computed. However, it is too challenging to tackle the value function $V(\bm{x})$ in \eqref{vf} directly, since it involves the manipulation of trajectories for DSPS over the infinite time horizon. Therefore, we reduce it to the unique bounded solution to the Bellman type equation \eqref{eq} when $\alpha\in (0,1)$ and, thus, the problem of computing the maximal robust invariant set is 
reduced to solving the equation \eqref{eq}. The link between the value function \eqref{vf} and \eqref{eq} is established via the following dynamic programming principle.

\begin{lemma}
\label{dp}
 For $\bm{x}\in \mathbb{R}^n$ and $l\in\mathbb{N}$, we have:
\begin{equation}
\label{dp1}
\begin{split}
V(\bm{x})=\sup_{ \pi \in \Pi_{\bm{x}}}\max&\big\{\alpha^l V(\bm{\phi}_{\bm{x}}^{\pi}(l)),\sup_{i\in [0,l)\cap\mathbb{N}}\max_{j\in \{1,\ldots,n_0\}} \alpha^i  h'_{j}(\bm{\phi}_{\bm{x}}^{\pi}(i))\big\}.
\end{split}
\end{equation}
\end{lemma}
\begin{proof}
Let 
\begin{equation}
\begin{split}
W(l,\bm{x}):=\sup_{ \pi \in \Pi_{\bm{x}}}&\max\big\{\alpha^l V(\bm{\phi}_{\bm{x}}^{\pi}(l)), \sup_{i\in [0,l)\cap \mathbb{N}}\max_{j\in \{1,\ldots,n_0\}} \alpha^i h'_{j}(\bm{\phi}_{\bm{x}}^{\pi}(i))\big\}.
\end{split}
\end{equation}
We will prove that for $\epsilon>0$, $|W(l,\bm{x})-V(\bm{x})|<\epsilon$.

According to the definition of $V(\bm{x})$, i.e., \eqref{vf}, for any $\epsilon_1>0$, there exists a disturbance policy  $\pi'=(\bm{d}'(i))_{i\in \mathbb{N}}\in \Pi_{\bm{x}}$ such that
\[V(\bm{x})\leq \sup_{i\in \mathbb{N}}\max_{j\in \{1,\ldots,n_0\}}\{\alpha^i h'_{j}(\bm{\phi}_{\bm{x}}^{\pi'}(i))\}+\epsilon_1.\]

We then introduce two disturbance policies $\pi_1=(\bm{d}_1(i))_{i\in \mathbb{N}} \in \Pi_{\bm{x}}$ and $\pi_2=(\bm{d}_2(i))_{i\in \mathbb{N}}\in \Pi_{\bm{y}}$ with $\bm{d}_1(j)=\bm{d}'(j)$ for $j=0,\ldots,l-1$ and $\bm{d}_2(j) =\bm{d}'(j+l)$ for $j\in \mathbb{N}$ respectively,  where $\bm{y}=\bm{\phi}_{\bm{x}}^{\pi_1}(l)$.  Thus we obtain that
\begin{equation*}
\label{w}
\begin{split}
W(l,\bm{x})&\geq\max\big\{\alpha^l V(\bm{y}),\sup_{i\in [0,l)\cap \mathbb{N}}\max_{j\in \{1,\ldots,n_0\}} \alpha^i h'_{j}(\bm{\phi}_{\bm{x}}^{\pi_1}(i))\big\}\\
&\geq \max\Big\{\sup_{i\in [l,+\infty)\cap \mathbb{N}}\max_{j\in \{1,\ldots,n_0\}}\{\alpha^i h'_{j}(\bm{\phi}_{\bm{y}}^{\pi_2}(i-l))\},\sup_{i\in [0,l)\cap\mathbb{N}}\max_{j\in \{1,\ldots,n_0\}}\{\alpha^i h'_{j}(\bm{\phi}_{\bm{x}}^{\pi_1}(i))\}\Big\}\\
&=\max\Big\{\sup_{i\in [l,+\infty)\cap \mathbb{N}}\max_{j\in \{1,\ldots,n_0\}}\{\alpha^i  h'_{j}(\bm{\phi}_{\bm{x}}^{\pi'}(i))\}, \sup_{i\in [0,l)\cap \mathbb{N}}\max_{j\in \{1,\ldots,n_0\}} \{\alpha^i h'_{j}(\bm{\phi}_{\bm{x}}^{\pi'}(i))\}\Big\}\\
&=\sup_{i\in \mathbb{N}}\max_{j\in \{1,\ldots,n_0\}}\{\alpha^i  h'_{j}(\bm{\phi}_{\bm{x}}^{\pi'}(i))\}\\
&\geq V(\bm{x})-\epsilon_1.
\end{split}
\end{equation*}
Therefore,
\begin{equation}
\label{geq}
V(\bm{x})\leq W(l, \bm{x})+\epsilon_1.
\end{equation}

On the other hand, by the definition of $W(l,\bm{x})$, for every $\epsilon_1>0$, there exists $\pi_1=(\bm{d}_1(i))_{i\in \mathbb{N}}\in \Pi_{\bm{x}}$ 
such that
\begin{equation}
\begin{split}
W(l,\bm{x})\leq& \max\big\{\alpha^l V(\bm{\phi}_{\bm{x}}^{\pi_1}(l)),\sup_{i\in [0,l)\cap \mathbb{N}}\max_{j\in \{1,\ldots,n_0\}} \{\alpha^i  h'_{j}(\bm{\phi}_{\bm{x}}^{\pi_1}(i))\}\big\}+\epsilon_1.
\end{split}
\end{equation}
Also, by the definition of $V(\bm{x})$, i.e., \eqref{vf}, for every $\epsilon_1>0$, there exists $\pi_2=(\bm{d}_2(i))_{i\in \mathbb{N}}\in \Pi_{\bm{y}}$ such that
\[ V(\bm{y})\leq \sup_{i\in \mathbb{N}}\max_{j\in \{1,\ldots,n_0\}} \{\alpha^i h'_{j}(\bm{\phi}_{\bm{y}}^{\pi_2}(i))\}+\epsilon_1,\]
where $\bm{y}=\bm{\phi}_{\bm{x}}^{\pi_1}(l)$.
  We define $ \pi=(\bm{d}(i))_{i\in \mathbb{N}}$ such that $\bm{d}(i) = \bm{d}_1(i)$ for $i=0,\ldots,l-1$ and $\bm{d}(i+l)=\bm{d}_2(i)$ for $i\in \mathbb{N}$. Obviously, $\pi \in \Pi_{\bm{x}}$.
Then, it follows
\begin{equation*}
\begin{split}
W(l, \bm{x})\leq &\max\Big\{\sup_{i\in \mathbb{N}\cap [l,\infty)}\max_{j\in \{1,\ldots,n_0\}} \{\alpha^i  h'_{j}(\bm{\phi}_{\bm{y}}^{\pi_2}(i-l))\},\sup_{i\in [0,l)\cap \mathbb{N}}\max_{j\in \{1,\ldots,n_0\}}\{\alpha^i  h'_{j}(\bm{\phi}_{\bm{x}}^{\pi_1}(i))\}\Big\}+2\epsilon_1\\
\leq &\sup_{i\in [0,+\infty)\cap \mathbb{N}}\max_{j\in \{1,\ldots,n_0\}} \{\alpha^i h'_{j}(\bm{\phi}_{\bm{x}}^{\pi}(i))\}+2\epsilon_1\\
\leq &V(\bm{x}) +2\epsilon_1.
\end{split}
\end{equation*}
Therefore, 
\begin{equation}
\label{leq}
V(\bm{x})\geq W(l,\bm{x})-2\epsilon_1.
\end{equation}

Combining \eqref{geq} and \eqref{leq}, we finally have $|V(\bm{x})-W(l,\bm{x})|\leq \epsilon=2\epsilon_1$. Since $\epsilon_1$ is arbitrary, $V(\bm{x})=W(l,\bm{x})$ holds for $\bm{x}\in \mathbb{R}^n$ and $l\in \mathbb{N}$. This completes the proof. 
\end{proof}

Based on Lemma \ref{dp}, we derive \eqref{eq}, to which $V(\bm{x})$ is a unique bounded solution when $\alpha\in (0,1)$.
\begin{theorem}
\label{equations}
If $\alpha\in (0,1]$, the value function $V(\bm{x}):\mathbb{R}^n\rightarrow \mathbb{R}$ in \eqref{vf} is a bounded solution to the Bellman type equation \eqref{eq}. Moreover, $V(\bm{x})$ is the unique bounded solution to \eqref{eq} when $\alpha\in (0,1)$.
\end{theorem}
\begin{proof}
When $l=1$, \eqref{dp1} is reduced to 
\begin{equation*}
\begin{split}
    V(\bm{x})=\sup_{\pi \in \Pi_{\bm{x}}}\max \big\{&\alpha V(\bm{\phi}_{\bm{x}}^{\pi}(1)),\sup_{i\in [0,1)\cap\mathbb{N}}\max_{j\in \{1,\ldots,n_0\}} \alpha^i  h'_{j}(\bm{\phi}_{\bm{x}}^{\pi}(i))\big\}, 
    \end{split}
\end{equation*}
which is further equivalent to 
\begin{equation*}
\begin{split}
    V(\bm{x})=\sup_{\bm{d} \in D}\max \big\{\alpha V(\bm{f}(\bm{x},\bm{d})), \max_{j\in \{1,\ldots,n_0\}} h'_{j}(\bm{x})\big\}.
    \end{split}
\end{equation*}
Thus, \eqref{eq} is the special case of \eqref{dp1} when $l=1$. In the rest we just prove the statement that $V(\bm{x})$ is the unique bounded solution to \eqref{eq} when $\alpha\in (0,1)$.

Assume that $U(\bm{x})$ is a bounded solution to \eqref{eq} as well, and there exists $\bm{y}\in \mathbb{R}^n$ such that $U(\bm{y})\neq V(\bm{y})$. Without loss of generality, we assume that $U(\bm{y})<V(\bm{y})$, i.e., there exists $\delta>0$ such that $V(\bm{y})-U(\bm{y})=\delta$. 

Since $U(\bm{y})-\max_{j\in \{1,\ldots,n_0\}} h'_{j}(\bm{y})\geq 0$, $$V(\bm{y})-\max_{j\in \{1,\ldots,n_0\}} h'_{j}(\bm{y})>0$$ holds. Consequently, we have that $$V(\bm{y})=\alpha \sup_{\bm{d}\in D}V(\bm{f}(\bm{y},\bm{d})).$$ Also, due to the fact that $U(\bm{y})\geq \alpha \sup_{\bm{d}\in D}U(\bm{f}(\bm{y},\bm{d}))$, $$\alpha \sup_{\bm{d}\in D}V(\bm{f}(\bm{y},\bm{d}))-\alpha\sup_{\bm{d}\in D}U(\bm{f}(\bm{y},\bm{d}))\geq \delta$$ holds, implying that 
$$\alpha \sup_{\bm{d}\in D}(V(\bm{f}(\bm{y},\bm{d}))-U(\bm{f}(\bm{y},\bm{d})))\geq \delta.$$
Therefore, for $1<\beta<\frac{1}{\alpha}$, there exists $\bm{d}\in D$ such that 
$$\alpha (V(\bm{f}(\bm{y},\bm{d}))-U(\bm{f}(\bm{y},\bm{d})))\geq \beta \alpha \delta.$$
 Let $\bm{d}_1$ satisfy $$V(\bm{f}(\bm{y},\bm{d}_1))-U(\bm{f}(\bm{y},\bm{d}_1))\geq \beta \delta,$$ and $\bm{y}_1=\bm{f}(\bm{y},\bm{d}_1)$. It is obvious that $$V(\bm{y}_1)-U(\bm{y}_1)\geq \beta \delta.$$ Repeating the above procedure, we can construct a sequence $(\bm{y}_j)_{j=1}^{\infty}$ satisfying $V(\bm{y}_j)-U(\bm{y}_j)\geq \beta^j \delta$. Thus, $$V(\bm{y}_j)\geq U(\bm{y}_j)+\beta^j \delta, \forall j\in \mathbb{N}.$$ 
 Since $U$ is bounded over $\mathbb{R}^n$ and $$\lim_{j\rightarrow \infty}\beta^j \delta=\infty,$$ we have that $V(\bm{y}_j)$ approaches infinity when $j$ tends to infinity, contradicting that $V$ is bounded over $\mathbb{R}^n$. Therefore, $V(\bm{y})=U(\bm{y})$. Based on the above deduction technique, a similar contradiction can be obtained for the case when $U(\bm{y})>V(\bm{y})$. 

Thus, this concludes that the value function $V(\bm{x}):\mathbb{R}^n\rightarrow \mathbb{R}$ in \eqref{vf} is the unique bounded solution to \eqref{eq}. 
\end{proof}

\begin{remark}
In this paper we take the function $h'_j(x)$ of the particular form $\frac{h_{j}(\bm{x})}{1+h_{j}^2(\bm{x})}$, $j=1,\ldots,n_0$. Actually, 
it can take various forms such as $\frac{h_{j}(\bm{x})}{1+|h_{j}(\bm{x})|}$ which can still make Proposition \ref{sets}, Lemma \ref{dp}
and Theorem \ref{equations} hold, but it should satisfy the following properties:
\begin{enumerate}
\item $\{\bm{x}\mid \bigwedge_{j=1}^{n_0} h'_j(\bm{x})\leq 0\}=\{\bm{x}\mid \bigwedge_{j=1}^{n_0} h_j(\bm{x})\leq 0\}=X$;
\item $h'_j(\bm{x})$ is bounded over $\mathbb{R}^n$, $j=1,\ldots, n_0$. 
\end{enumerate}
The difference with various $h'_j(\bm{x})$'s in computing  the maximal robust invariant set will be investigated as the future work.
\end{remark}

From Theorem \ref{equations}, we conclude that the maximal robust invariant set $\mathcal{R}_0$ can be obtained by solving \eqref{eq}. A technique for solving \eqref{eq} with $\alpha \in (0,1)$ is the value iteration in the framework of reinforcement learning.
\begin{theorem}
\label{iteration}
Suppose the sequence of functions $(V_i(\bm{x}))_{i\in \mathbb{N}}$ with $V_i(\cdot):\mathbb{R}^n\rightarrow \mathbb{R}$ is generated by the value iteration starting from some bounded function $V_0(\cdot):\mathbb{R}^n\rightarrow \mathbb{R}$ according to 
\begin{equation}
\label{iter}
\begin{split}
V_{i+1}(\bm{x})=\max\big\{\sup_{\bm{d}\in D}\alpha V_i(\bm{f}(\bm{x},\bm{d})), \max_{j\in \{1,\ldots,n_0\}}h'_{j}(\bm{x})\big\}
\end{split}
\end{equation}
for $\bm{x}\in \mathbb{R}^n$ and $i\in \mathbb{N}$,
then $V_i(\bm{x})$ uniformly approximates $V(\bm{x})$ over $\mathbb{R}^n$ if $\alpha\in (0,1)$ as $i$ tends to infinity, where $V(\bm{x})$ is the unique bounded solution to \eqref{eq}. 
\end{theorem} 
\begin{proof}
According to \eqref{iter}, we have 
\begin{equation}
\label{iter1}
\begin{split}
&V_{i+1}(\bm{x})-V_{i}(\bm{x})\\
&= \max\Big\{\alpha \sup_{\bm{d}_i\in D} V_i\big(\bm{f}(\bm{x},\bm{d}_i)\big),\max_{j\in \{1,\ldots,n_0\}}h'_{j}(\bm{x})\Big\}-\max\Big\{\alpha \sup_{\bm{d}_i\in D} V_{i-1}\big(\bm{f}(\bm{x},\bm{d}_i)\big),\max_{j\in \{1,\ldots,n_0\}}h'_{j}(\bm{x})\Big\}\\
&\leq \max\Big\{\alpha \sup_{\bm{d}_i\in D}\Big(V_i\big(\bm{f}(\bm{x},\bm{d}_i)\big)-V_{i-1}\big(\bm{f}(\bm{x},\bm{d}_i)\big)\Big),0\Big\}\\
&\leq \max\Big\{\alpha^{i}\sup_{\bm{d}_1\in D}\cdots\sup_{\bm{d}_i\in D}\Big(V_1\big(\bm{g}(\bm{x},\bm{d}_1,\cdots,\bm{d}_i)\big)-V_{0}\big(\bm{g}(\bm{x},\bm{d}_1,\cdots,\bm{d}_i)\big)\Big),0\Big\}
\end{split}
\end{equation}
and 
\begin{equation}
\label{iter2}
\begin{split}
&V_{i+1}(\bm{x})-V_{i}(\bm{x})\\
&= -\min\Big\{\alpha \inf_{\bm{d}_i\in D}-V_i\big(\bm{f}(\bm{x},\bm{d}_i)\big),-\max_{j\in \{1,\ldots,n_0\}}h'_{j}(\bm{x})\Big\}+\min\Big\{\alpha \inf_{\bm{d}_i\in D} -V_{i-1}\big(\bm{f}(\bm{x},\bm{d}_i)\big),-\max_{j\in \{1,\ldots,n_0\}}h'_{j}(\bm{x})\Big\}\\
&\geq \min\Big\{\alpha \inf_{\bm{d}_i\in D}\Big(V_i\big(\bm{f}(\bm{x},\bm{d}_i)\big)-V_{i-1}\big(\bm{f}(\bm{x},\bm{d}_i)\big)\Big),0\Big\}\\
&\geq \min\Big\{\alpha^{i}\inf_{\bm{d}_1\in D}\cdots \inf_{\bm{d}_i\in D}\Big(V_1\big(\bm{g}(\bm{x},\bm{d}_1,\cdots,\bm{d}_i)\big)-V_{0}\big(\bm{g}(\bm{x},\bm{d}_1,\cdots,\bm{d}_i)\big)\Big),0\Big\},
\end{split}
\end{equation}
 where \[\bm{g}(\bm{x},\bm{d}_1,\cdots,\bm{d}_i)=\underbrace{\bm{f}\bigg(\bm{f}\Big(\cdots\bm{f}\big( \bm{f}}_i(\bm{x},\bm{d}_i),\bm{d}_{i-1}\big),\cdots,\bm{d}_2\Big),\bm{d}_1\bigg).\] 
Moreover, since $V_0$ and $\max_{j\in \{1,\ldots,n_0\}}h'_{j}$ are bounded over $\mathbb{R}^n$, therefore, $V_1$ is bounded as well. Thus, according to \eqref{iter1},  \eqref{iter2} and $\alpha\in (0,1)$, we have that $V_i(\bm{x})$ uniformly approximates a function $V'(\bm{x})$ over $\mathbb{R}^n$ as $i$ tends to infinity. In the rest  we just need to prove that $V'(\bm{x})=V(\bm{x})$ over $\bm{x}\in \mathbb{R}^n$. This conclusion can be assured by replacing $V_{i+1}(\bm{x})-V_i(\bm{x})$ in \eqref{iter1} and \eqref{iter2} with $V_{i+1}(\bm{x})-V(\bm{x})$, resulting in that $V_{i+1}(\bm{x})$ uniformly approximates $V(\bm{x})$ over $\mathbb{R}^n$ as $i$ tends to infinity, where $V(\bm{x})=\max\{\alpha \sup_{\bm{d}\in D} V(\bm{f}(\bm{x},\bm{d})),\max_{j\in \{1,\ldots,n_0\}}h'_j(\bm{x})\}$.
\end{proof}

\begin{remark}
\label{rate1}
From Theorem \ref{iteration}, we observe that the initial function $V_0(\bm{x})$ for the value iteration \eqref{iter} can be an arbitrary bounded function from $\mathbb{R}^n$ to $\mathbb{R}$. Let $|V_0(\bm{x})|\leq M$ for $\bm{x}\in \mathbb{R}^n$ and $\alpha\in (0,1)$, where $M\geq 0$.
From \eqref{iter1} and \eqref{iter2} we can obtain that
\begin{equation}
\label{rate}
 \sup_{\bm{x}\in \mathbb{R}^n}|V_{i+1}(\bm{x})-V_i(\bm{x})|\leq \alpha^i \max \{2M,1+M\}.
 \end{equation}
Consequently, $V_i(\bm{x})$ converges to $V(\bm{x})$  over $\mathbb{R}^n$ with  the rate of convergence $\alpha$. This also  implies that the smaller  the rate of convergence $\alpha$ is, the faster the convergence is. 

In addition, \eqref{rate} indicates that a smaller $M$ results in a faster convergence to the unique bounded solution to \eqref{eq}. Thus, $M=0$, i.e., $V_0(\bm{x})\equiv 0$ for $\bm{x}\in \mathbb{R}^n$, is the best choice. This claim holds when $V(\bm{x})$ in \eqref{vf} is unknown. Due to the fact that 
$V(\bm{x})=\max\{\alpha \sup_{\bm{d}\in D} V(\bm{f}(\bm{x},\bm{d})),\max_{j\in \{1,\ldots,n_0\}}h'_j(\bm{x})\}$, we have that $V_0(\bm{x})=V(\bm{x})$ for $\bm{x}\in \mathbb{R}^n$ is the best choice if $V(\bm{x})$ is known.
\end{remark}

The  value iteration for addressing \eqref{eq} with $\alpha\in (0,1)$ is described in Alg. \ref{alg1}.
\begin{algorithm}
\caption{The Value Iteration for Solving \eqref{eq}}
\begin{enumerate}
\item Set $V_0(\bm{x}):=0$ over $\bm{x}\in \mathbb{R}^n$ and $l:=0$, and decide on a grid $\Lambda=\{\bm{x}_1,\ldots,\bm{x}_N\}$ on $X$ for the state variable $\bm{x}$, and a grid $\Delta=\{\bm{d}_1,\ldots,\bm{d}_M\}$ in $D$ for the disturbance variable $\bm{d}$.
\item Choose a value in $(0,1)$ for $\alpha$;
\item Choose an accuracy tolerance $\epsilon>0$;
\item For each $\bm{x}_i\in \Lambda$, $i=1,\ldots,N$, compute 
$$\bm{x}'_{i,j}=\bm{f}(\bm{x}_i,\bm{d}_j),\forall j=1,\ldots, M,$$ then compute an interpolated value function at each $\bm{x}'_{i,j}: \tilde{V}_l(\bm{x}'_{i,j})$ and compute 
$$V_{l+1}(\bm{x}_i)=\max\{\alpha \max_{j\in \{1,\ldots,M\}}\tilde{V}_l(\bm{x}'_{i,j}),\max_{j\in \{1,\ldots,n_0\}}h'_{j}(\bm{x}_{i})\}.$$
\item If $\max_{\bm{x}\in \Lambda}|V_{l+1}(\bm{x})-V_l(\bm{x})|<\epsilon$, go to step 6); otherwise, $l:=l+1$ and go back to 3);
\item Obtain the final solution $V(\bm{x})$ as 
 $V(\bm{x})\approx V_{l+1}(\bm{x})$.
\end{enumerate}
\label{alg1}
\end{algorithm}

\begin{remark}
\label{Iter_no}
From \eqref{rate} we can obtain that given the accuracy tolerance $\epsilon$,  the maximum iteration number in Alg. \ref{alg1} is predictable a priori, i.e., the maximum iteration number is less than or equal to $\max\{\lceil \log_{\alpha} \epsilon \rceil,1\}$, where $\lceil \log_{\alpha} \epsilon \rceil$ is the smallest integer being larger than or equal to $\log_{\alpha} \epsilon$.
\end{remark}

The termination of the value iteration in Alg. \ref{alg1} is guaranteed by Theorem \ref{iteration}, which is also uncovered in Remark \ref{Iter_no}. In this way the value of $V(\bm{x})$ can be calculated on the grid points of the set $X$. However, this is a computationally expensive process and as the size of the state and disturbance spaces grows, it becomes intractable generally. 

When $\alpha=1$, we cannot guarantee the convergence of $(V_{i})_{i\in \mathbb{N}}$ in  \eqref{iter}. This is also reflected in Remark \ref{rate1}. Even if the sequence $(V_{i})_{i\in \mathbb{N}}$ converges, the convergence to $V(\bm{x})$ is not guaranteed, where $V(\bm{x})$ is the value function in \eqref{vf}, since \eqref{eq} may not have a unique bounded solution. However, the equation \eqref{eq} with $\alpha=1$ will enable us to construct a semi-definite program for synthesizing robust invariant sets efficiently.

\subsection{Semi-definite Programming Implementation}
\label{AA}
When the dimensions of the product space of state and disturbance spaces are appropriate, the value iteration described in Alg. \ref{alg1} could be employed to solve \eqref{eq} for estimating the maximal robust invariant set. However, it suffers severely from the curse of dimensionality, consequently highlighting the need of alternative numerical methods. Similar to  \cite{xue2018}, we in this section present a semi-definite program derived from  \eqref{eq} for synthesizing robust invariant sets. For the sake of completeness and ease of understanding, we also give an appropriate description on how to construct the semi-definite program.

From \eqref{eq} we observe that if a function $u(\bm{x}):\mathbb{R}^n\rightarrow \mathbb{R}$ satisfies \eqref{eq}, then it satisfies
\begin{equation}
\label{upper1}
\left\{
\begin{array}{lll}
u(\bm{x})-\alpha u(\bm{f}(\bm{x},\bm{d}))\geq 0, \forall \bm{d}\in D, \forall \bm{x}\in \mathbb{R}^n,\\
u(\bm{x})-h'_{j}(\bm{x})\geq 0, \forall \bm{x}\in \mathbb{R}^n,\forall j\in \{1,\ldots,n_0\}.
\end{array}
\right.
\end{equation}

\begin{corollary}
\label{upper}
For any function $u(\bm{x}): \mathbb{R}^n \rightarrow \mathbb{R}$ satisfying \eqref{upper1}, $\{\bm{x}\in \mathbb{R}^n\mid u(\bm{x})\leq 0\}$ is a robust invariant set when $\alpha\in (0,1]$.  Furthermore, if $\alpha\in (0,1)$, $u(\bm{x})\geq 0$ over $\bm{x}\in \mathbb{R}^n$ and consequently $$\{\bm{x}\in \mathbb{R}^n\mid u(\bm{x})\leq 0\}=\{\bm{x}\in \mathbb{R}^n\mid u(\bm{x})=0\}.$$
\end{corollary}
\begin{proof}
The statement that $\{\bm{x}\in \mathbb{R}^n\mid u(\bm{x})\leq 0\}$ is a robust invariant set when $\alpha\in (0,1]$ can be justified by following the proof of Corollary 1 in \cite{xue2018}.

In the sequel we prove that $u(\bm{x})\geq 0$ over $\bm{x}\in \mathbb{R}^n$ when $\alpha\in (0,1)$. Assume that there exists $\bm{y}_1\in \mathbb{R}^n$ such that $u(\bm{y}_1)<0$. Due to the fact that $$u(\bm{f}(\bm{y}_1,\bm{d}_1))\leq \frac{1}{\alpha}u(\bm{y}_1), \forall \bm{d}_1\in D,$$ we obtain that 
$u(\bm{g}(\bm{y}_1,\bm{d}_1,\ldots,\bm{d}_i))\leq \frac{1}{\alpha^i}u(\bm{y}_1),\forall \bm{d}_1,\ldots, \forall \bm{d}_i \in D,$ where \[\bm{g}(\bm{y}_1,\bm{d}_1,\ldots,\bm{d}_i)=\underbrace{\bm{f}\bigg(\bm{f}\Big(\cdots\bm{f}\big(\bm{f}}_i(\bm{y}_1,\bm{d}_1),\bm{d}_2\big),\cdots,\bm{d}_{i-1}\Big),\bm{d}_i\bigg).\] Therefore, we have that $$\lim_{i\rightarrow \infty}u(\bm{g}(\bm{y}_1,\bm{d}_1,\ldots,\bm{d}_i))=-\infty,$$ contradicting $$u(\bm{x})\geq \max_{j\in \{1,\ldots,n_0\}}h'_{j}(\bm{x}), \forall \bm{x}\in \mathbb{R}^n.$$ Consequently, $u(\bm{x})\geq 0$ for $\bm{x}\in \mathbb{R}^n$ when $\alpha\in (0,1)$. Therefore, $\{\bm{x}\in \mathbb{R}^n\mid u(\bm{x})\leq 0\}=\{\bm{x}\in \mathbb{R}^n\mid u(\bm{x})=0\}$ when $\alpha\in (0,1)$. An immediate result is $\{\bm{x}\in \mathbb{R}^n\mid u(\bm{x})\leq 0\}=\{\bm{x}\in \mathbb{R}^n\mid u(\bm{x})=0\}$.
 \end{proof}

Since \eqref{upper1} has indicator functions on the expression $u(\bm{x}_0)-u(\bm{f}(\bm{x}_0,\bm{d}))$, which is beyond the capability of the solvers we use. We would like to first obtain a constraint by removing indicators according to Lemma \ref{split}.
\begin{lemma}[\cite{ChenHWZ15}]
\label{split}
Suppose $\bm{F}(\bm{x})=1_{Y_1}\cdot\bm{F}_1(\bm{x})+\cdots+1_{Y_{k'}}\cdot \bm{F}_{k'}(\bm{x})$ and $\bm{G}(\bm{x})=1_{Z_1}\cdot \bm{G}_1(\bm{x})+\cdots+1_{Z_{l'}}\cdot \bm{G}_{l'}(\bm{x})$, where $\bm{x}\in \mathbb{R}^n$, $k',l'\in \mathbb{N}$, and $Y_i, Z_j\subseteq \mathbb{R}^n$, $i=1,\ldots,k'$, $j=1,\ldots,l'$. Also, $Y_1,\ldots, Y_{k'}$ and $Z_1,\ldots,Z_{l'}$ are respectively disjoint. Then, $\bm{F}\leq \bm{G}$ if and only if (pointwise)
\begin{equation}
\begin{split}
&\bigwedge_{i=1}^{k'}\bigwedge_{j=1}^{l'}\big[Y_i\wedge Z_j\Rightarrow \bm{F}_i\leq \bm{G}_j\big]\wedge\\
 &\quad\quad\quad\quad\bigwedge_{i=1}^{k'} \big[Y_i\wedge \big(\bigwedge_{j=1}^{l'} \neg Z_j\big)\Rightarrow \bm{F}_i\leq 0\big]\wedge\\
& \quad\quad\quad\quad\quad\quad\quad\quad\bigwedge_{j=1}^{l'}\big[\big(\bigwedge_{i=1}^{k'}\neg Y_i\big)\wedge Z_j\Rightarrow 0\leq \bm{G}_j\big].
\end{split}
\end{equation}
\end{lemma}

Based on Lemma \ref{split}, the equivalent constraint without indicator functions of \eqref{upper1} is formulated below:
\begin{equation}
\label{upper2}
\begin{split}
&\bigwedge_{i=1}^{k}\big[u(\bm{x})-\alpha u( \bm{f}_i(\bm{x},\bm{d}))\geq 0, \forall \bm{d}\in D_i,~\forall \bm{x}\in X_i\big]\wedge \\
&\bigwedge_{j=1}^{n_0}\big[u(\bm{x})-h'_{j}(\bm{x})\geq 0, \forall \bm{x}\in \mathbb{R}^n\big].
\end{split}
\end{equation}

Assume that $\Omega(X)$ is the set of states, which are reachable from the set $X$ within one step computation, i.e., $$\Omega(X)=\{\bm{x}\mid \bm{x}=\bm{f}(\bm{x}_0,\bm{d}),\bm{x}_0\in X,\bm{d}\in D\}\cup X,$$ which can be obtained by solving semi-definite programs or linear programs \cite{lasserre2015,magron2017}. Herein, we assume that it was already given.
In addition, we define $\sum[\bm{y}]$ to be the set of sum of squares polynomials over variables $\bm{y}$, i.e.,
\[\sum[\bm{y}]=\{p\in \mathbb{R}[\bm{y}]\mid p=\sum_{i=1}^k q_i^2,q_i\in \mathbb{R}[\bm{y}],i=1,\ldots,k\}.\]
When $u(\bm{x})$ in \eqref{upper1} is constrained to polynomial type and is restricted in a ball $$B=\{\bm{x} \in \mathbb{R}^n\mid g(\bm{x})\geq 0\}$$ with $g(\bm{x})= R-\sum_{i=1}^n x_i^2$ such that $\Omega(X)\subseteq B$, \eqref{upper1} is relaxed as a semi-definite program \eqref{sos}.
\begin{algorithm}
\caption{Semi-definite Programming Implementation for Computing Robust Invariant Sets}
\begin{equation}
\label{sos}
\begin{split}
&\min_{u,s_{j}^{X_i},s_{j}^{D_i},s_{i},s'_{j}} \bm{c} \cdot \bm{w}\\
&u(\bm{x})-\alpha u(\bm{f}_i(\bm{x},\bm{d}))+\sum_{j=1}^{n_i}s^{X_i}_{j}h_{i,j}(\bm{x})+\sum_{j=1}^{m'_{i}}s^{D_i}_{j}p_{i,j}(\bm{d})-s_{i}g(\bm{x}) \in \sum[\bm{x},\bm{d}],\\
&(1+h_{j}^2)u(\bm{x})-h_{j}(\bm{x})-s'_{j}g(\bm{x})\in \sum[\bm{x}],\\
&i=1,\ldots,k; j=1,\ldots,n_0,
\end{split}
\end{equation}
where $\bm{c}\cdot \bm{w}=\int_{B}ud\bm{x}$, $\bm{w}$ is the constant vector computed by integrating the monomials in $u(\bm{x})$ over $B$, $\bm{c}$ is the vector composed of unknown coffecients in $u(\bm{x})$, $s^{X_i}_{j}, s_{j}^{D_i}, s_{i}\in \sum[\bm{x},\bm{d}]$ and $s'_{j}\in \sum[\bm{x}]$.  
\label{alg2}
\end{algorithm}

The problem of solving \eqref{sos} can be reformulated as a semi-definite programming problem, which falls within the convex programming framework. Note that the objective of \eqref{sos} facilitates the gain of  less conservative robust invariant sets. The reason is that if $\psi_1(\bm{x})\leq \psi_2(\bm{x})$ over $B$, then $\{\bm{x}\in B\mid \psi_2(\bm{x})\leq 0\}\subseteq \{\bm{x}\in B\mid \psi_1(\bm{x})\leq 0\}$ and $\int_{B} \psi_1(\bm{x})d\bm{x}\leq \int_{B} \psi_2(\bm{x})d\bm{x}$.

\begin{theorem}
\label{inner}
Let $u(\bm{x})\in \mathbb{R}[\bm{x}]$ be a solution to \eqref{sos}, then $\{\bm{x}\in B\mid u(\bm{x})\leq 0\}$ is a robust invariant set. Furthermore, if $\alpha\in (0,1)$, $u(\bm{x})\geq 0$ over $\bm{x}\in \mathbb{R}^n$ and consequently $\{\bm{x}\in B\mid u(\bm{x})\leq 0\}=\{\bm{x}\in B\mid u(\bm{x})=0\}.$
\end{theorem}
\begin{proof}
Since $u(\bm{x})$ satisfies the constraint in \eqref{sos} and $\bm{f}(\bm{x},\bm{d})$ satisfies \eqref{path_1}, we obtain that $u(\bm{x})$ satisfies according to $\mathcal{S}-$ procedure presented in  \cite{boyd1994} and Lemma \ref{split} 
\begin{align}
&u(\bm{x})-\alpha u(\bm{f}(\bm{x},\bm{d}))\geq 0, \forall \bm{d}\in D, \forall \bm{x}\in  B \text{ and }\label{1} \\
&(1+h_{j}^2(\bm{x}))u(\bm{x})-h_{j}(\bm{x})\geq 0, \forall j\in \{1,\ldots,n_{0}\}.
\label{2}
\end{align}

Assume that there exist an initial state $\bm{y}\in \{\bm{x}\in B\mid u(\bm{x})\leq 0\}$ and a disturbance policy $\pi' \in \Pi_{\bm{y}}$ such that $\bm{\phi}_{\bm{y}}^{\pi'}(l)\in X$ does not hold for every $ l\in \mathbb{N}$. Since \eqref{2} holds, we have the conclusion that $$\{\bm{x}\in B\mid u(\bm{x})\leq 0\} \subset X$$ and thus $\bm{y}\in X$. Let $l_0\in \mathbb{N}$ be the first time making $\bm{\phi}_{\bm{y}}^{\pi'}(l)$ stay outside the set $X$, i.e. $$\bm{\phi}_{\bm{y}}^{\pi'}(l_0)\in B\setminus X \text{~and~}\bm{\phi}_{\bm{y}}^{\pi'}(l)\in X$$ for $l=0,\ldots,l_0-1$. That is, $u(\bm{\phi}_{\bm{y}}^{\pi'}(l_0))>0$. However, since $\Omega(X)\subset B$, \eqref{1} and \eqref{2}, we derive that $$u(\bm{\phi}_{\bm{y}}^{\pi'}(l_0))\leq 0.$$ This is a contradiction. Thus, every possible trajectory to DSPS initialized in $\{\bm{x}\in B\mid u(\bm{x})\leq 0\}$ does not leave the safe set $X$. Therefore, $$\{\bm{x}\in B\mid u(\bm{x})\leq 0\}$$ is a robust invariant set. 

When $\alpha \in (0,1)$,  using analogous arguments as in the proof of Corollary \ref{upper}, we obtain $u(\bm{x})\geq 0$ over $\{\bm{x}\in B\mid u(\bm{x})\leq 0\}$ and as a result we have $$\{\bm{x}\in B\mid u(\bm{x})\leq 0\}=\{\bm{x}\in B\mid u(\bm{x})=0\}.$$
\end{proof}

Theorem \ref{inner} indicates that a robust invariant set can be represented by $\{\bm{x}\in B\mid u(\bm{x})=0\}$ when $\alpha\in (0,1)$. However, such typical representation often results in extremely conservative robust invariant sets. Therefore, we generally assign $1$ to $\alpha$ when employing the semi-definite program \eqref{sos} for synthesizing robust invariant sets. This statement is similar to Remark 1 in \cite{WZF19}.

\section{Experiments}
\label{ex}
In this section we evaluate the performance of both the value iteration described in Alg. \ref{alg1}  and the semi-definite program \eqref{sos} described in Alg. \ref{alg2} on three illustrative examples. Moreover, we compare the methods in this paper with the ones in \cite{xue2018} based on these three examples.

The parameters that control the performance of our methods are presented in Table \ref{table}. All computations were performed on an i7-7500U 2.70GHz CPU with 32GB RAM running Windows 10. For numerical implementation of the semi-definite program \eqref{sos}, we formulate the sum of squares problem \eqref{sos} using the MATLAB package YALMIP\footnote{It can be downloaded from \url{https://yalmip.github.io/}.} \cite{lofberg2004} and use Mosek\footnote{The software Mosek can be obtained free from \url{https://www.mosek.com/} for academic use.} \cite{mosek2015mosek} as a semi-definite programming solver. For the value iteration, uniform grids are adopted for state and disturbance spaces. The state spaces for Example \ref{ex0}, \ref{ex2} and \ref{ex3} are respectively restricted to $[-1.1,1.1]\times [-1.1,1.1]$, $[-1,1]\times [-1,1]$ and $[-1,1]^7$.

\begin{table*}[t]
\begin{center}
\begin{tabular}{|c|c|c|c|c|c||c|c|c|c|c|c|c|c|c|}
  \hline
  \multirow{1}{*}{}&\multicolumn{5}{|c||}{\texttt{SDP}} & \multicolumn{5}{|c|}{\texttt{VI}} \\\hline
   Ex.&$d_u$&$d_{s}$&$d_{s'}$&$\alpha$&$T_{\mathtt{SDP}}$&$\alpha$&$\epsilon$& N& M&$T_{\mathtt{VI}}$\\\hline
   1&10&10&10&$1$&1.06 & {0.01}&$10^{-20}$&$10^4$&10& 81.50 \\\hline
                           
   \multirow{2}{*}{2}&8&10&10 &1&4.55 &\multirow{2}{*}{0.01}&\multirow{2}{*}{$10^{-20}$}&\multirow{2}{*}{$10^4$}&\multirow{2}{*}{10}&\multirow{2}{*}{307.32}\\
                          
                              &12&18&18&1&161.43& & && &\\\hline                 
         3& 4&4 &4 &1& 242.86& 0.01&$10^{-10}$ &$10^7$&10&--\\\hline            
   \end{tabular}
\end{center}
\caption{\textit{Parameters and performance of our implementations on the Examples \ref{ex0}$\sim$\ref{ex3}. $\alpha$: the parameter value in \eqref{vf}; $d_u, d_{s}, d_{s'}$: the degree of the polynomials $u, \{s_{i}, s_{j}^{X_i}, s_{j}^{D_i}\}$ and $s'_{j}$ in \eqref{sos}, respectively; $T_{\mathtt{SDP}}$: computation times (seconds) in solving \eqref{sos}; $\epsilon$: the stopping criterion in the value iteration in Alg. \ref{alg1}; $N,M$: numbers of elements in $\Lambda$ and $\Delta$ respectively in the value iteration in Alg. \ref{alg1}; $T_{\mathtt{VI}}$: computation times (seconds) in solving \eqref{eq} using the value iteration in Alg. \ref{alg1}.} }
\label{table}
\end{table*}

\begin{example}
\label{ex0}
Consider the discrete-generation predator-prey model from \cite{halanay2000stability} without switching, i.e., there is only one subsystem, where $$\bm{f}(x,y)=(ax-bxy;-cy+dxy)$$ with $a=0.5$, $b=1$, $c=0.5$ and $d=1$, $X=\{(x,y)\mid  x^2+y^2-1\leq 0\}$. 

We first solved \eqref{eq} for computing an approximation of the maximal robust invariant set based on the value iteration in Alg. \ref{alg1}. The obtained approximation is illustrated in Fig. \ref{fig-one-1}.  The level sets of the corresponding computed solution to \eqref{eq} are illustrated in Fig. \ref{fig-one-0}.
We also used the semi-definite programming based method in Alg. \ref{alg2} to synthesize robust invariant sets.  The result when $d_u=10$ is also visualized in Fig. \ref{fig-one-1}. 

\begin{figure}
\centering
\includegraphics[width=3.0in,height=2.00in]{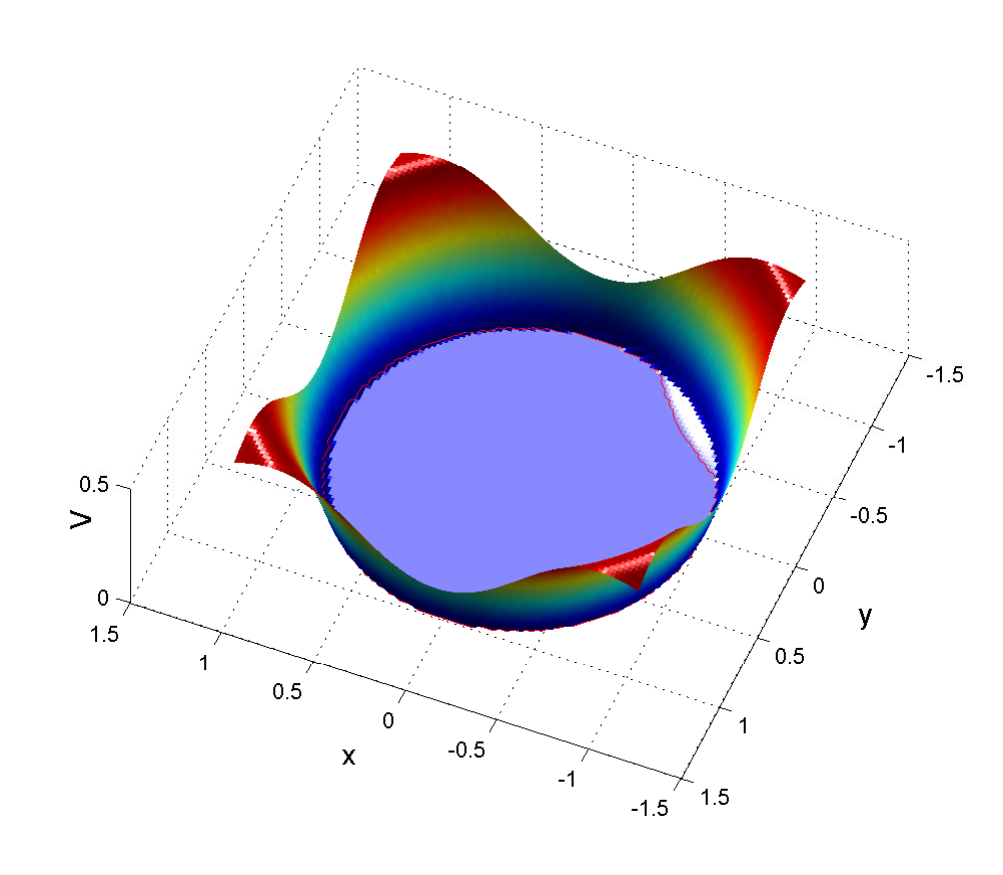}
\caption{Level sets of $V$ computed via the value iteration in Alg. \ref{alg1} for Example \ref{ex0}.}
\label{fig-one-0}
\end{figure}

\begin{figure}[!h]
\center
\includegraphics[width=3.0in,height=2.0in]{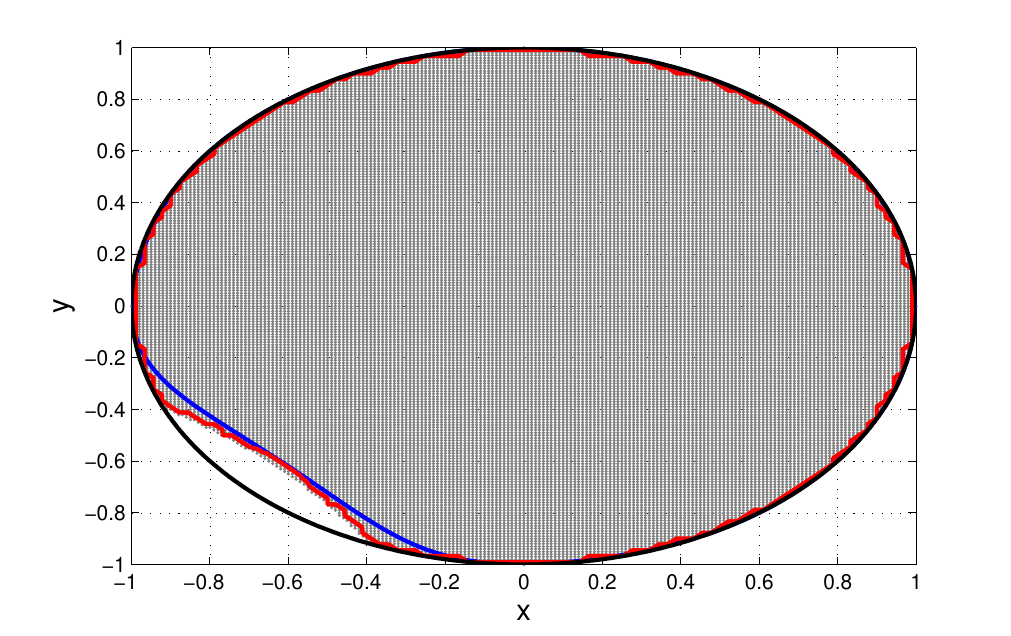}
\caption{Computed robust invariant sets for Example \ref{ex0}. (Blue curve denotes the boundary of the robust invariant set computed via Alg. \ref{alg2} with $d_u=10$.  Red curve denotes the boundary of the maximal robust invariant set computed via the value iteration in Alg. \ref{alg1}. Gray region denotes the maximal robust invariant set estimated via numerical simulation techniques. Black curve denotes the boundary of the safe set $X$.)}
\label{fig-one-1}
\end{figure}
\end{example}

\begin{example}
\label{ex2}
We consider a \emph{DSPS} with 
\begin{equation*}
\begin{split}
&\bm{f}_1(x,y)=(x;(0.5+d)x-0.1y),\\
&\bm{f}_2(x,y)=(y;0.2x-(0.1+d)y+y^2),
\end{split}
\end{equation*}
$X=\{(x,y)\mid  x^2+y^2-0.8\leq 0\}$, $X_1=\{(x,y)\mid 1-(x-1)^2-y^2\leq 0\}$, $X_2=\{(x,y)\mid -1+(x-1)^2+y^2<0\}$, $D_1=\{d\mid d^2-0.01\leq 0\}$ and $D_2=\{d\mid d^2-0.01\leq 0\}$. 

The robust invariant sets computed by solving \eqref{sos} when $d_u=8$ and $12$ are illustrated in Fig. \ref{fig-one-4}. Fig. \ref{fig-one-4} also presents the maximal robust invariant set, which is computed via  the value iteration in Alg. \ref{alg1}. The level sets of the corresponding computed solution to \eqref{eq} are shown in Fig. \ref{fig-one-3}. 

\begin{figure}[!h]
\centering
\includegraphics[width=3.0in,height=2.0in]{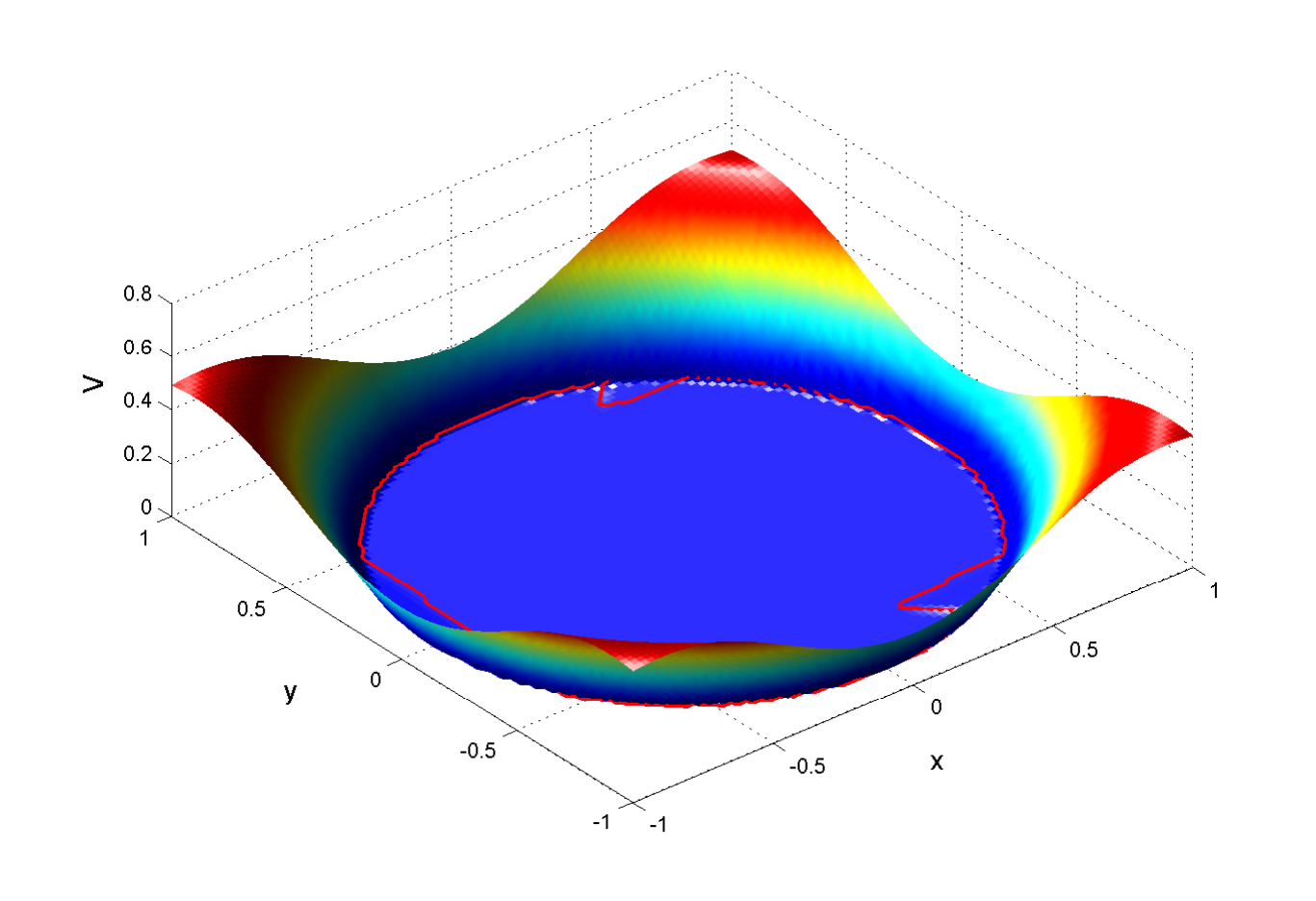}
\caption{Level sets of $V$ obtained via the value iteration in Alg. \ref{alg1} for Example \ref{ex2}.}
\label{fig-one-3}
\end{figure}

\begin{figure}[!h]
\center
\includegraphics[width=3.0in,height=1.5in]{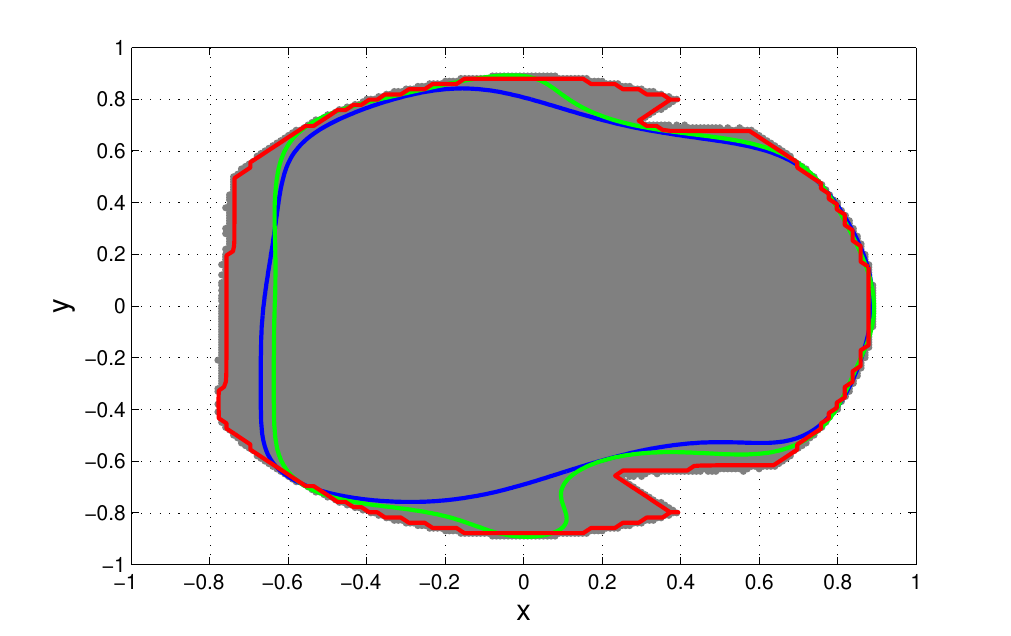}
\caption{Computed robust invariant sets for Example \ref{ex2}. \small{Blue and green curves denote the boundaries of robust invariant sets computed via Alg. \ref{alg2} with $d_u=8$ and $12$, respectively. Red curve denotes the boundary of the maximal robust invariant set computed via the value iteration in Alg. \ref{alg1}. Gray region denotes the maximal robust invariant set estimated via numerical simulation techniques.}}
\label{fig-one-4}
\end{figure}

\end{example}

The plots in Fig. \ref{fig-one-0} and Fig. \ref{fig-one-3} confirm the statement in Proposition \ref{sets} that the solution to \eqref{eq} with $\alpha\in (0,1)$ is non-negative. In addition, the semi-definite program \eqref{sos} with $\alpha=0.5$ is applied to Examples \ref{ex0} and \ref{ex2} as well. However, we did not obtain non-empty robust invariant sets for both examples based on the parameters in Table \ref{table}. This justifies the last statement in Subsection \ref{AA}.

The semi-definite programming based method in Alg. \ref{alg2} computes an inner-approximation of the maximal robust invariant set and consequently  brings conservativeness in estimating the maximal robust invariant set. This statement can be further justified by the visualized results in Fig. \ref{fig-one-1} and \ref{fig-one-4}. In contrast,  the value iteration in Alg. \ref{alg1} can solve the Bellman type equation \eqref{eq} with $\alpha\in (0,1)$ and produce an approximation of the maximal robust invariant set. However,  it requires partitioning the state and disturbance spaces, thereby exhibiting exponential growth in computational complexity with the number of state and disturbance variables and preventing its application to systems of dimension being larger than five. We illustrate this issue through an example with seven state variables.

\begin{example}
\label{ex3}
We consider a \emph{DSPS} with seven dimensional variables $\bm{x}=(x_1,x_2,x_3,x_4,x_5,x_6,x_7)$. In this example, 
\begin{equation*}
\begin{split}
    &\bm{f}_1(\bm{x})=((0.5+d)x_1;0.8x_2;0.6x_3+0.1x_6;x_4; 0.8x_5;0.1x_2+x_6;0.2x_2+0.6x_7),\\
    &\bm{f}_2(\bm{x})=(0.5x_1+0.1x_6;(0.5+d)x_2;x_3;0.1x_1+0.4x_4;0.2x_1+x_5;x_6;0.1x_1+x_7),
\end{split}
\end{equation*}
 $X=\{\bm{x}\mid  \sum_{i=1}^7x_i^2-1\leq 0\}$, $X_1=\{\bm{x}\mid x_1+x_2+x_3-x_4-x_5-x_6-x_7\geq 0\}$, $X_2=\{\bm{x}\mid x_1+x_2+x_3-x_4-x_5-x_6-x_7<0\}$, and $D_1=D_2=\{d \in \mathbb{R}\mid d^2-0.01\leq 0\}$. 

 Unlike for the low-dimensional Examples \ref{ex0} and \ref{ex2}, the value iteration for solving \eqref{eq} here runs out of memory and thus does not return a result. The semi-definite programming based method \eqref{sos} in Alg. \ref{alg2}, however, is still able to compute robust invariant sets, which are illustrated in Fig. \ref{fig-one-5}. In order to gauge the accuracy of computed robust invariant sets, we use numerical simulation techniques to obtain coarse approximations of the maximal robust invariant sets on planes $x_5-x_6$ with $x_1=x_2=x_3=x_4=x_7=0$ and $x_6-x_7$ with $x_1=x_2=x_3=x_4=x_5=0$ respectively. These approximations are also depicted 
in Fig. \ref{fig-one-5}.

\begin{figure}[!h]
\centering
\includegraphics[width=3.0in,height=1.5in]{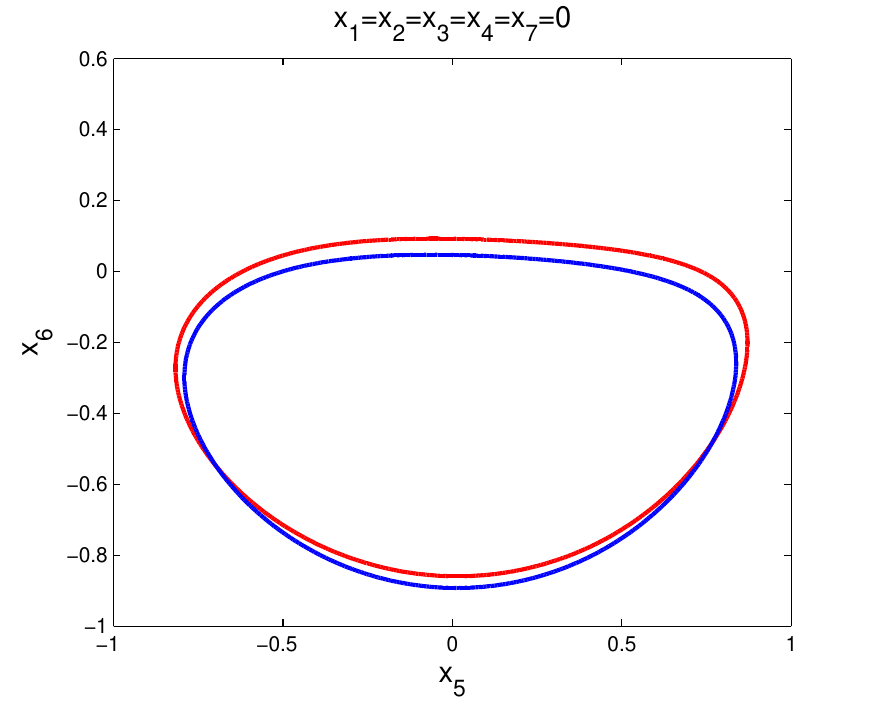}\\
\includegraphics[width=3.0in,height=1.5in]{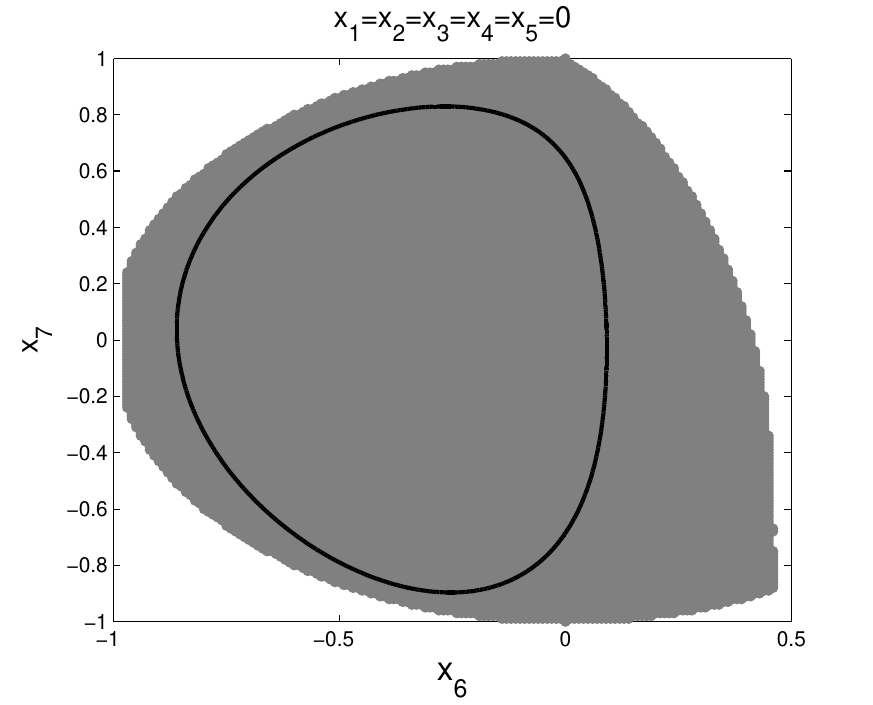}
\caption{Computed robust invariant sets for Example \ref{ex3}. Black curve denotes the boundary of the robust invariant set computed via Alg. \ref{alg2} with $d_u=4$. Gray region is the maximal robust invariant set estimated via numerical simulation techniques.}
\label{fig-one-5}
\end{figure}
\end{example}

Besides, we compare the methods in the present paper with the ones in \cite{xue2018} based on Examples \ref{ex0}$\sim$\ref{ex3}. We firstly compare the Bellman type equation based method in this paper with the one in \cite{xue2018}. Based on the same parameters listed in Table \ref{table}, the value iteration to solve the Bellman type equation in \cite{xue2018} does not terminate after one hour for both Examples \ref{ex0} and \ref{ex2}. This is because the value iteration for solving the Bellman type equation in \cite{xue2018} may not converge. Consequently, it can not be used to compute an approximation of the maximal robust invariant set. 
Similar to the one in this paper, the value iteration for solving the Bellman type equation in \cite{xue2018} runs out of memory and thus does not return an estimate for Example \ref{ex3}. Next, we compare the semi-definite programming based method \eqref{sos} with the one in \cite{xue2018} based on the same parameters in Table \ref{table}. The results obtained by these two methods for Examples \ref{ex0}, \ref{ex2} and \ref{ex3} are respectively visualized in Fig. \ref{fig-one-51}, \ref{fig-one-6} and \ref{fig-one-7}. We from Fig. \ref{fig-one-51} observe that almost the same robust invariant sets are produced for Example \ref{ex0} by these two methods\footnote{`almost the same' means that the computed maximal robust invariant sets showing on the graph are indistinguishable.}. Fig. \ref{fig-one-6} indicates that \eqref{sos} produces less conservative robust invariant sets for Example \ref{ex2} than that in \cite{xue2018} for both $d_u=8$ and $d_u=12$. Based on  Fig. \ref{fig-one-7}, it is hard to tell which results are more conservative and thus hard to compare the performance on Example \ref{ex3}. However, these comparisons indicate that the semi-definite program \eqref{sos} definitely achieves a reasonable improvement  over the one in \cite{xue2018} on estimating robust invariant sets for some cases, e.g. Example \ref{ex2}. In order to give an insight into the merits of one over the other for these two semi-definite programming formulations, the investigation of their structure differences  would be necessary. We will investigate this in our future work.

\begin{figure}[!h]
\center
\includegraphics[width=3.0in,height=1.3in]{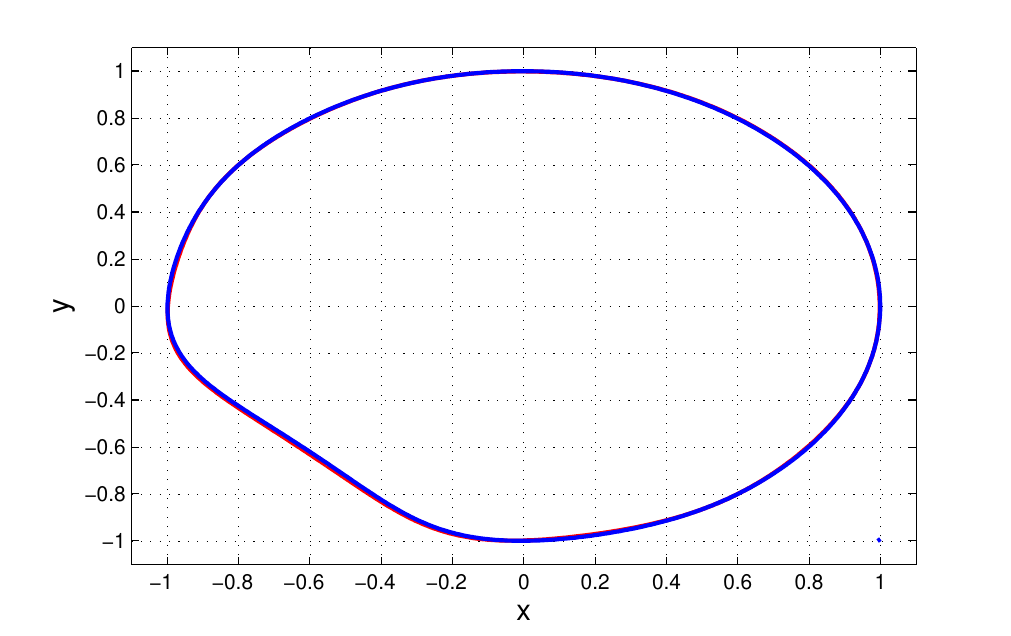}
\caption{Robust invariant sets for Example \ref{ex0} computed when $d_u=10$.  Blue and Red curves denote boundaries of robust invariant sets computed via \eqref{sos} and the semi-definite program in \cite{xue2018}, respectively.}
\label{fig-one-51}
\end{figure}

\begin{figure}[!h]
\center
\includegraphics[width=3.0in,height=1.3in]{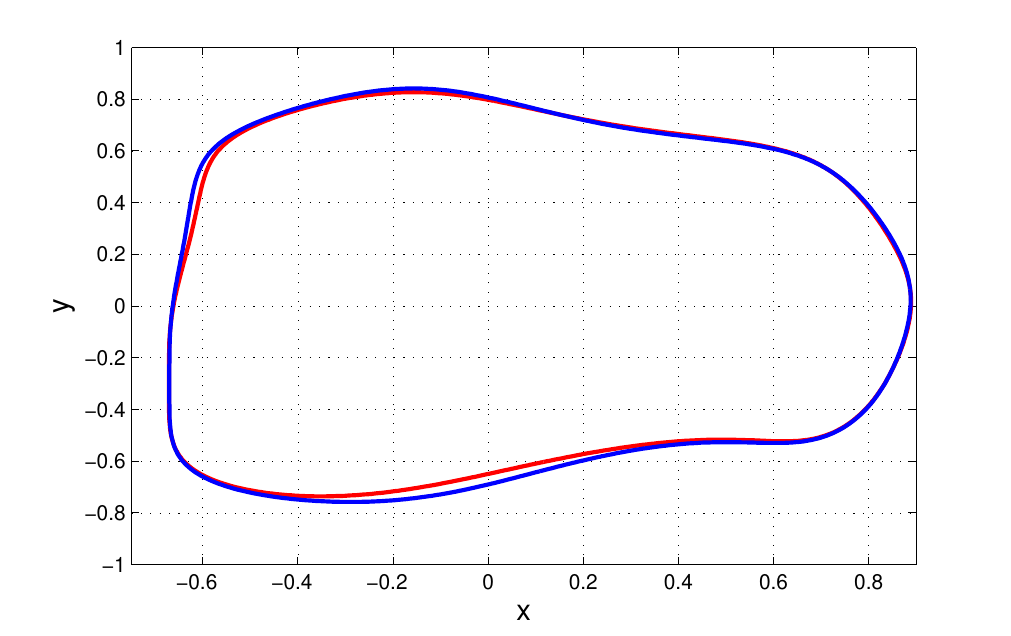}\\
\includegraphics[width=3.0in,height=1.3in]{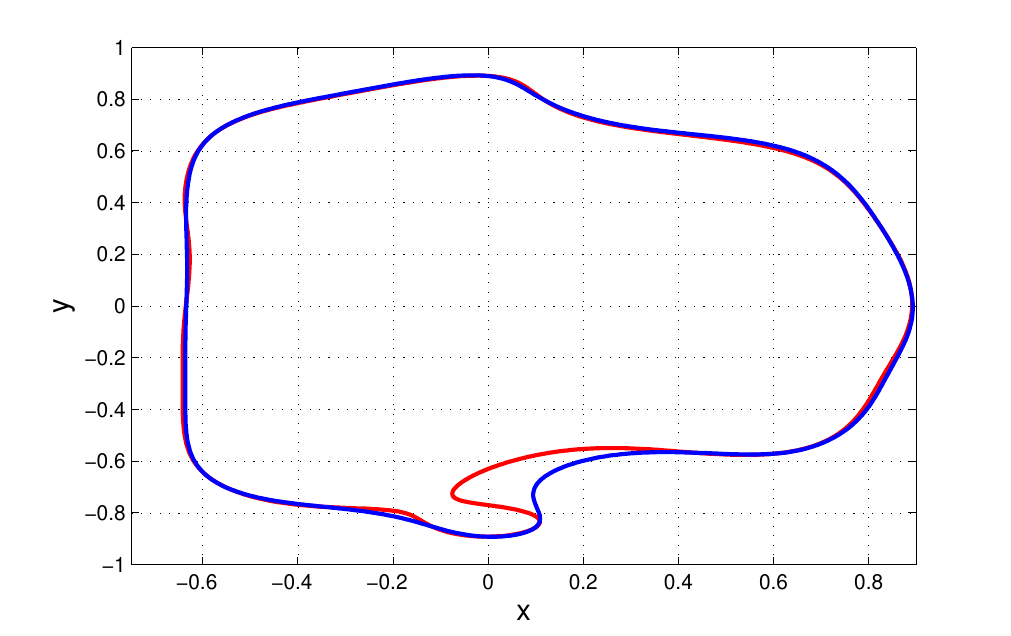}
\caption{Robust invariant sets for Example \ref{ex2} computed when $d_u=8, 12$ (from above to below). Blue and Red curves denote the boundaries of robust invariant sets computed via \eqref{sos} and the semi-definite program in \cite{xue2018}, respectively.}
\label{fig-one-6}
\end{figure}

\begin{figure}[!h]
\center
\includegraphics[width=3.0in,height=1.3in]{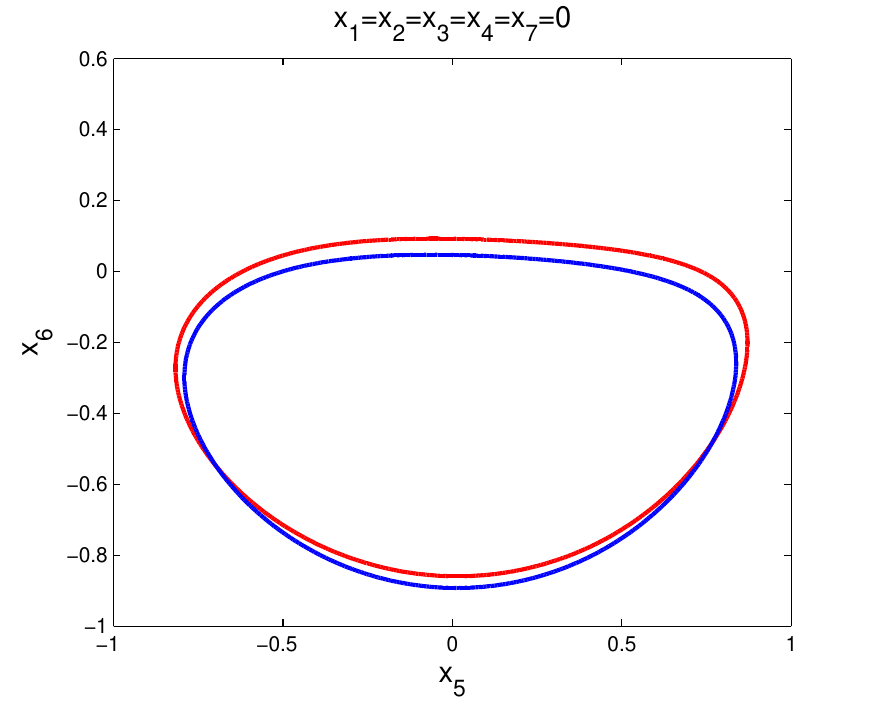}\\
\includegraphics[width=3.0in,height=1.3in]{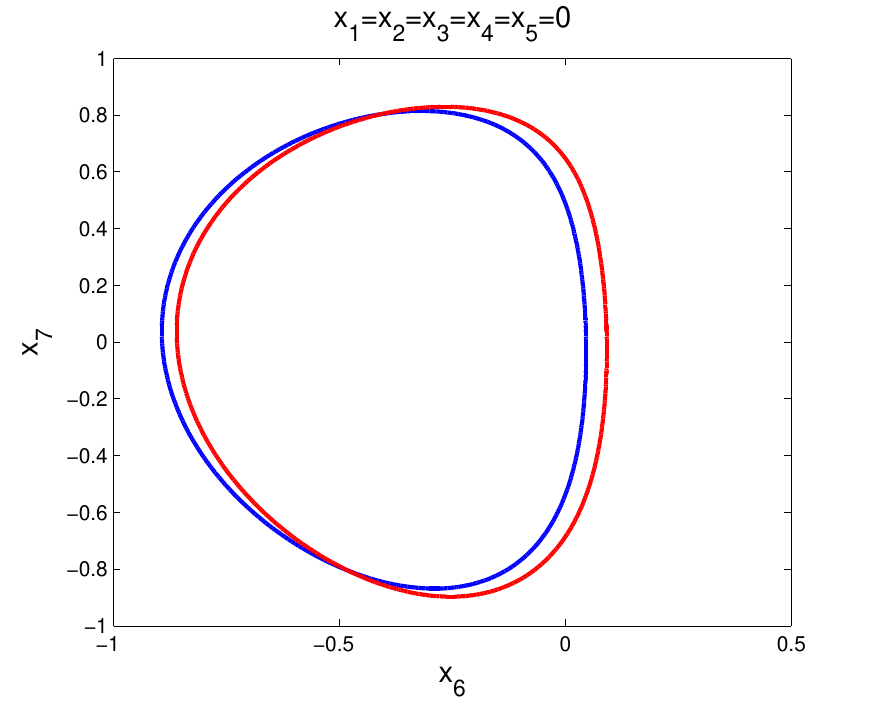}
\caption{Robust invariant sets for Example \ref{ex3} computed when $d_u=4$. Blue and Red curves denote the boundaries of robust invariant sets computed via \eqref{sos} and the semi-definite program in \cite{xue2018}, respectively.}
\label{fig-one-7}
\end{figure}

\section{Conclusion and Future Work}
\label{con}
In this paper we studied the robust invariant sets generation problem for discrete-time switched polynomial systems subject to disturbance inputs. We for the first characterize the maximal robust invariant set as the zero level set of the unique bounded solution to a Bellman type equation. Value iteration can be used to solve such equation with appropriate number of state and disturbance variables. Besides, based on the derived Bellman type equation a semi-definite program was also implemented to synthesize robust invariant sets. Three examples demonstrated the performance of our methods. 

In near future we would like to compare our method with the interval branch-and-bound approach in \cite{li2018invariance}. Also, we would extend our methods to the computation of robust invariant sets for  state-constrained hybrid systems subject to competing inputs (control and perturbation).

\bibliographystyle{abbrv}
\bibliography{reference}
\end{document}